\documentclass[runningheads]{llncs}

\usepackage{graphicx,amssymb}

\usepackage{graphicx,amssymb}
\usepackage[linesnumbered,ruled,vlined]{algorithm2e}
\usepackage{amsmath,amssymb}

\usepackage[usenames]{color}
\usepackage{longtable}
\usepackage{booktabs}
\usepackage{url}
\usepackage{mathrsfs}
\usepackage{amsfonts}
\usepackage{marvosym}
\usepackage{lscape}
\usepackage{enumerate}
\usepackage{caption}
\usepackage{subcaption}
\usepackage{multirow}
\usepackage{verbatim}
\usepackage{pst-node}
\usepackage[pagewise]{lineno}
\usepackage{makecell}
\newcommand{\gr}[1]{\mathcal{#1}}
\newcommand{\sset}[1]{\mathcal{#1}}
\DeclareMathOperator*{\argmax}{\arg\!\max}
\DeclareMathOperator*{\argmin}{\arg\!\min}
\def\polylog{\operatorname{polylog}}

\begin{document}
\title{Extracting Densest Sub-hypergraph with Convex Edge-weight Functions \thanks{This work is partially supported by the National Natural Science Foundation of China under Grant No. 61802049.}}
%
%
\author{Yi Zhou \inst{1,2} \orcidID{0000-0002-9023-4374} \and
Shan Hu \inst{1}\and
Zimo Sheng\inst{1}}
\authorrunning{Y. Zhou et al.}
%

\institute{School of Computer Science and Engineering, University of Electronic Science and Technology of China, Chengdu 611731, China \and
Zhejiang Jinggong Steel Building Group, Shaoxing  312030, Zhejiang, China
\\
\email{zhou.yi@uestc.edu.cn},
\email{hu.shan@std.uestc.edu.cn}, 
\email{1491858607@qq.com}
}

\maketitle              
\begin{abstract}

The densest subgraph problem (DSG) aiming at finding an induced subgraph such that  the average edge-weights of the subgraph is maximized, is a well-studied problem.
However, when the input graph is a hypergraph, the existing notion of DSG fails to capture the fact that a hyperedge partially belonging to an induced sub-hypergraph is also a part of the  sub-hypergraph.
To resolve the issue, we suggest a function $f_e:\mathbb{Z}_{\ge0}\rightarrow \mathbb{R}_{\ge 0}$ to represent the partial edge-weight of a hyperedge $e$ in the input  hypergraph $\gr{H}=(V,\sset{E},f)$ and formulate a \emph{generalized densest sub-hypergraph problem} (GDSH) as $\max_{S\subseteq V}\frac{\sum_{e\in \sset{E}}{f_e(|e\cap S|)}}{|S|}$. We demonstrate that, when all the edge-weight functions are non-decreasing convex, GDSH can be solved in polynomial-time by the linear program-based algorithm, the network flow-based algorithm and the greedy $\frac{1}{r}$-approximation algorithm where $r$ is the rank of the input hypergraph.
Finally, we investigate the computational tractability of GDSH where some edge-weight functions are non-convex. 
\keywords{Densest subgraph problem \and Hypergraph \and Convex function.}
\end{abstract}

%
%
%
\section{Introduction}
\label{sec_introduction}
The \emph{densest subgraph problem} (DSG) is a well-known problem in research communities of operations research, combinatorial optimization, data mining and so on.
Given an edge-weighted graph $\gr{G}=(V, \sset{E}, w)$ with a vertex set $V$, an edge set $\sset{E}$ and an edge-weight function $w:\sset{E}\rightarrow \mathbb{R}_{\ge 0}$, DSG asks us to maximize the density of the subgraph induced by a vertex set $S\subseteq V$, i.e., $\max_{S\subseteq V}\frac{\sum_{e\in \sset{E}}{w(e)}}{|S|}$. 
Applications of DSG range from web community detection \cite{dourisboure2007extraction,chen2010dense}, network motif clustering \cite{benson2016higher,tsourakakis2015k} to information recommendation \cite{rahman2017density}.
For solving DSG, there exists a network flow-based exact algorithm by Goldberg \cite{goldberg_finding_1984,chekuri2022densest}, a linear program-based algorithm by Charikar \cite{charikar2000greedy} and a linear-time $\frac{1}{2}$-approximation algorithm in \cite{asahiro2000greedily,khuller2009finding}.

On the other hand, \emph{hypergraph} is attracting increasing attentions in recent years.
The hypergraph is a generalization of the normal graph in which a hyperedge consists of arbitrary positive number of vertices.
An edge-weighted hypergraph is defined as $\mathcal{H}=(V,\sset{E}, w)$ where $V$ is a vertex set, $\sset{E}$ is  a hyperedge set, and $w:\sset{E}\rightarrow \mathbb{R}_{\ge 0}$ is an edge-weight function assigning each hyperedge a positive weight.
The densest subgraph problem in an edge-weighted hypergraph, i.e. the \emph{densest sub-hypergraph problem} (DSH), is known to be formulated as follows.

\begin{problem}[DSH as in \cite{tsourakakis2015k,hu2017maintaining}]
\label{problem_old}
Given a hypergraph $\gr{H}=(V,\sset{E}, w)$,  DSH asks for a sub-hypergraph induced by $S\subseteq V$ such that the average edge-weight of the sub-hypergraph i.e., $\frac{\sum_{e\subseteq S}w(e)}{|S|}$, is maximized.
\end{problem}

As far as we know, DSH was initially appeared as a generalization of the \emph{densest $r$-clique problem (D$r$C)} in \cite{tsourakakis2015k}. 
Given a graph $\gr{G}=(V,\sset{E})$, D$r$C asks for a subset of vertices $S$ such that the average number of $r$-cliques (a clique of size $r$) induced by $S$ is maximized. 
Clearly, D$r$C can be reduced to DSH by building a $r$-uniform hypergraph in which each hyperedge represents a $r$-clique.
For D$r$C, a polynomial exact algorithm and a $\frac{1}{r}$-approximation algorithm were introduced in \cite{tsourakakis2015k}, and a sampling algorithm was given  in \cite{mitzenmacher2015scalable}.
In \cite{hu2017maintaining}, Hu \emph{et al} finally remove the assumption that the input hypergraph is $r$-uniform and formalize DSH as Problem \ref{problem_old} .
They demonstrated that the linear program, network flow and approximation algorithms for DSH.
Recently, a much faster $(1-\epsilon)$ approximation algorithm based on max flow is given in \cite{chekuri2022densest} for DSH.


In this paper, we study a more generalized version of the densest sub-hypergraph problem rather than continue working with the existing model.
We observed that in Problem \ref{problem_old}, a hyperedge $e$ is counted as a part of sub-hypergraph induced by $S$ only when $e$ is a subset of $S$. 
However, in some graph applications like \cite{li2017inhomogeneous}, if a hyperedge $e$ intersects with $S$, i.e., $e\not\subset S$ and $e\cap S\neq \emptyset$, $e$ is  \emph{partially} belong to the sub-hypergraph induced by $S$.
Therefore, the fact that the weight of a sub-hypergraph induced by $S$ should contain a partial weight of the hyperedge that intersects with $S$ is not captured by the definition of DSH.
In order to fix this issue, we introduce an edge-weight function $f_e: \{0,...,|e|\} \rightarrow \mathbb{R}_{\ge 0}$ for each hyperedge $e$ in the input hypergraph, and then define the following \emph{generalized densest sub-hypergraph problem} (GDSH).

\begin{problem}[GDSH in this paper]
\label{problem_new}
Given an edge-weighted hypergraph $\mathcal{H}=(V,\sset{E}, f)$ with vertex set $V$ and hypergedge set $\sset{E}$, $f_e: \{0,...,|e|\} \rightarrow \mathbb{R}_{\ge 0}$ being an edge-weight function for each $e\in \sset{E}$, GDSH asks for a set of vertices $S\subseteq V$ such that i.e., $\frac{\sum_{e\in \sset{E}}f_e(|e\cap S|)}{|S|}$, is maximized.
\end{problem}

It is clear that GDSH generalizes the DSH problem.
For example, if $\mathcal{H}=(V,\sset{E},w)$ is the input hypergraph for DSH,
we can build a hypegraph $\mathcal{H}'=(V,\sset{E},f')$ such that $f'_e(i)=0$ when $i<|e|$ and $f'_e(|e|)=w(e)$. 
Then, it is clear that the solution of GDSH with input graph $\mathcal{H}'$ is the same as DSH with input graph $\mathcal{H}$. 
In this sense, GDSH also generalizes existing densest subgraph problems like DSG and D$r$CP.

Since  convex functions are ubiquitous in many applications, in the remaining of the paper, we investigate GDSH with focus on cases where all edge-weight functions are \emph{non-decreasing convex}. It is clear that all the above problems like DSG, D$r$C and DSH are special cases of GDSH with non-decreasing convex edge-weight functions. 
We will use $n$ to denote the vertex number $|V|$, $m$ to denote the edge number $|\sset{E}|$, $p$ to denote $\sum_{e\in \sset{E}}|e|$ and $r$ to denote the rank $\max_{e\in \sset{E}}|e|$ of the  input hypergraph $\mathcal{H}=(V,\sset{E}, f)$.
Note that $p=\sum_{v\in V}{\deg_\gr{H}(v)}$ where $\deg_\gr{H}(v)$ is the \emph{degree} of vertex $v\in V$. 
We also use $\Psi=\sum_{e\in \sset{E}}{f_e(|e|)}$ to denote the whole edge-weight of the hypergraph.
Our main contributions for GDSH when all edge-weight functions are non-decreasing convex functions are summarized as follows.

\begin{itemize}
    \item \textbf{A linear program whose optimal value is equal to the  maximum density of GDSH.}
    The linear program  has $O(mr!)$ inequalities but efficient oracles exist for separation. 
    We show that an optimal solution of GDSH can be easily obtained by solving the linear program.
    \item   \textbf{A network flow-based algorithm which runs in $O(mincut(p,pr)\log\Psi)$ time, $mincut(N,M)$ representing the time of solving minimum $s,t$-cut in directed flow network with $N$ vertices and $M$ arcs.}
    We also show the technique to obtain a  $O(mincut(p,pr)\log(\epsilon^{-1}\log(rm)))$ time $(1-\epsilon)$ approximation algorithm which removes the $\log\Psi$ factor. 
    \item \textbf{A greedy $\frac{1}{r}$ approximation algorithm  with much faster running time $O(pr\log n)$.} 
    With a little relaxation of the greedy strategy, the greedy approximation algorithm can also run in logarithmic iterations under the parallel computing settings.
\end{itemize}

It is worth mentioning that the above three algorithms extend the linear program algorithm, network flow algorithm and greedy algorithm, respectively, in \cite{hu2017maintaining,charikar2000greedy}.
However, the extension is not trivial as. We  only assume the non-decreasing and convexity properties of the edge-weight function in this work, contrary to the existing work that edge-weight functions are uniform and specifically given.
 
For completeness, we lastly study the computational tractability of GDSH when some edge-weight functions are non-convex.
It turns out that when all edge-weight functions are non-decreasing concave, GDSH can be simply solved by selecting a (densest) vertex,
when some edge-weight function are concave, GDSH is shown to be NP-hard by reduction from the max-cut problem.

\section{Properties of Edge-weight Functions}
\label{preliminaries}
Given $\gr{H}=(V,\sset{E}, f)$, the edge weight function $f_e$ is defined on discrete domain ${0,...,|e|}$. 
We first assume that  $f_e$  has \emph{non-decreasing} properties for any $e\in\sset{E}$.

\begin{property}[Non-decreasing]
\label{pro_nondecreasing}
$f_e(i)\le f_e(i+1), \forall i\in \{0,...,|e|-1\}$
\end{property}

This property is a clearly natural in practice.
Without loss of generality, we assume that $f_e(0)=0$. 
If $f_e(0)\neq 0$, we can use $f'_e(i)=f_e(i)-f_e(0)$ to replace $f_e$ without changing the optimal solution of GDSH.

Aside from the non-decreasing property, we also discuss the \emph{convexity} and \emph{concavity} properties. 
As we know, convexity and concavity are common properties for many functions. They play important roles in characterizing the hardness of underlying optimization problems.
\begin{property}[Convexity]
\[
f_e(i)-f_e(i-1) \le f_e(i+1)-f_e(i), \forall i\in \{1,...,|e|-1\}.
\]
\end{property}

\begin{property}[Concavity]
\[
f_e(i)-f_e(i-1) \ge f_e(i+1)-f_e(i), \forall i\in \{1,...,|e|-1\}.
\]
\end{property}

Given an $S\subseteq V$, we use $F(S)=\sum_{e\in \sset{E}}{f_e(|e\cap S|)}$ to represent the \emph{weight of sub-hypergraph induced by $S$}.
Clearly, if $\forall e\in \sset{E}$, $f_e$ is non-decreasing convex (concave), then $F(S)$ is a monotone \emph{supermodular} (\emph{submodular}) function in finite set $V$ (because  submodularity and supermodularity are closed under non-negative linear combination). 
Let us recall the definitions of \emph{supermodularity} and  \emph{submodularity} as bellow.
\begin{property}[Supmodularity]
\[
F(S\cup \{v\})-F(S) \le F(T\cup \{v\})-F(T), \forall S\subseteq T\subset 2^V \mbox{ and } \forall v\in S, v \notin T
\]
\end{property}

\begin{property}[Submodularity]
\[
F(S\cup \{v\})-F(S) \ge F(T\cup \{v\})-F(T), \forall S\subseteq T\subset 2^V \mbox{ and } \forall v\in S, v \notin T
\]
\end{property}

Lastly, we assume that $f_e(i)$ is computed in constant time for any $i\in \{0,...,|e|\}$.
Thus, for any set $S\subseteq V$, $f_e(|e\cap S|)$ is computed in time $O(min\{|e|,|S|\})$ and $F(S)$ is computed in time $O(p)$.


\section{GDSH with Convex Edge-weight Functions}
\label{convex}
In this section, we investigate algorithms for solving GDSH  when every edge-weight function is non-decreasing convex.
Specifically, we show a linear program, a parametric network flow-based algorithm, and a fast greedy approximation in Section \ref{sec_lp}, \ref{sec_flow} and  \ref{approx}, respectively.

\subsection{A Linear Program Approach}
\label{sec_lp}

For a hyperedge $e\in \sset{E}$, let $\sset{P}_e$  be the  set of all permutations of $e$. Given a permutation $\pi\in \sset{P}_e$, $\pi(i)=v$ means that the $i$th vertex of permutation $\pi$ is $v$ and $v\in e$. Then, the linear program for GDSH, i.e., LP-GDSH, is given as follows.

\begin{align}
\mbox{maximize~} &\displaystyle\sum_{e \in \sset{E}}{y_e} & \mbox{(LP-GDSH)} \nonumber\\
\mbox{s.t.~~~~} & \sum_{i=1}^{|e|}(f_e(i)-f_e(i-1))x_{\pi(i)} \ge  y_e & \forall e\in \sset{E}, \forall \pi\in \sset{P}_e \label{ineq1}\\
&\sum_{v\in V} x_v \le 1 & \label{ineq2}\\
&x_v\ge 0, y_e \ge 0 &\forall v \in V, e\in \sset{E} \nonumber
\end{align}

\begin{lemma}
\label{lemma_ineq_minimum}
Let $\mathbf{x}$ be a feasible solution of LP-GDSH.
Then, for any $e\in \sset{E}$, we have $\min_{\pi\in \sset{P}_e}(\sum_{i=1}^{|e|}(f_e(i)-f_e(i-1))x_{\pi(i)}) = \sum_{i=1}^{|e|}(f_e(i)-f_e(i-1))x_{\pi^*(i)}$ where $\pi^*\in \sset{P}_e$ is a permutation that $x_{\pi^*(1)}\ge x_{\pi^*(2)}\ge ...\ge x_{\pi^*(|e|)}$.
\end{lemma}
\begin{proof}
We justify the statement by contradiction.
Assume that $\pi' \in \sset{P}_e$ is a minimum permutation, i.e., $\sum_{i=1}^{|e|}(f_e(i)-f_e(i-1))x_{\pi'(i)}=min_{\pi\in\sset{P}_e}\sum_{i=1}^{|e|}(f_e(i)-f_e(i-1))x_{\pi(i)})$,  but there exists $i,j$ that $ 1\le i<j\le |e|, x_{\pi'(i)} < x_{\pi'(j)}$.
As $f_e$ is a non-decreasing convex function, we have $0\le f_e(i)-f_e(i-1)\le f_e(j)-f_e(j-1)$.
Then, we have $(f_e(i)-f_e(i-1))x_{\pi'(j)} + (f_e(j)-f_e(j-1))x_{\pi'(i)} < (f_e(i)-f_e(i-1))x_{\pi'(i)} + (f_e(j)-f_e(j-1))x_{\pi'(j)}$.
In other words, we can decrease $\sum_{i=1}^{|e|}(f_e(i)-f_e(i-1))x_{\pi'(i)}$ by exchanging $\pi'(i)$ and $\pi'(j)$,
which contradicts the assumption that $\pi'$ is the minimum permutation.
\end{proof}

\begin{theorem}
\label{thm_core_lp}
The following statements hold for LP-GDSH.
\begin{enumerate}
  \item For any $S\subseteq V$, there is  a feasible solution $(\mathbf{x,y})$ of LP-GDSH such that  $\sum_{e \in \sset{E}}{y_e}=\frac{F(S)}{|S|}$.
  \item  Let $\gamma^*$ be the optimal objective value  of LP-GDSH. Then, there is a vertex set $S\subseteq V$ such that $\gamma^*\le \frac{F(S)}{|S|}$.
\end{enumerate}
Therefore, the optimal solution of LP-GDSH is equal to the maximum density of GDSH.
\end{theorem}
\begin{proof}
\textbf{Proof of the first statement.} For any $S\subseteq V$, we construct a $\mathbf{x}$ such that $x_v=\frac{1}{|S|}$ if $v\in S$ and $x_v=0$ otherwise. Clearly, $\mathbf{x}$ satisfies Inequality \ref{ineq2}. We also construct $\mathbf{y}$ with $y_e=\frac{f_e(|e\cap S|)}{|S|}$ for any $e\in \sset{E}$. Then $\sum_{e \in \sset{E}}y_e$ is equal to $\frac{F(S)}{|S|}$. Now, let use  verify that this $(\mathbf{x},\mathbf{y})$ satisfies Inequality \ref{ineq1}. By Lemma \ref{lemma_ineq_minimum}, the left-hand side of Inequality \ref{ineq1} is at least  $\sum_{i=1}^{|e|}(f_e(i)-f_e(i-1))x_{\pi^*(i)}$ where $\pi^*\in {P}_e$ satisfies $\pi^*(i)>\pi^*(i+1)$ for $1\le i\le n-1$. Therefore,
\begin{align*}
&\sum_{i=1}^{|e|}(f_e(i)-f_e(i-1))x_{\pi(i)} \\
&\ge \sum_{i=1}^{|e|}(f_e(i)-f_e(i-1))x_{\pi^*(i)} \\
&=\sum_{i=1}^{|e\cap S|}{(f_e(i)-f_e(i-1))\frac{1}{|S|}} \\
&=\frac{f_e(|e\cap S|)}{|S|} \\
&=y_e.
\end{align*}
Hence, the first statement holds.

\textbf{Proof of the second statement.}
Let $(\mathbf{x}^*,\mathbf{y}^*)$ be an optimal solution of LP-GDSH.
Define $S_r=\{v:x^*_v\ge r\}$.
We claim that there exists $r\in[0,1]$ such that $\frac{F(S_r)}{|S_r|}\ge \gamma^*$.
Assume that there is no such $r$. Then we have $F(S_r)<\gamma^*|S_r|$ for any $r\in[0,1]$.
That is to say, $\int_{0}^{\infty}F(S_r)dr < \gamma^*\int_{0}^{\infty}|S_r|dr$.

On the other hand, we have
\[
\gamma^*\int_{0}^{1}|S_r|dr=\gamma^*\sum_{v\in V}{x^*_v}
\]
and
\begin{align*}
&\int_{0}^{1}F(S_r)dr =\sum_{e\in \sset{E}}\int_{0}^{1}f_e(|e\cap S|)dr \\
&=\sum_{e\in\sset{E}}( \sum_{i=1}^{|e|}(f_e(i)-f_e(i-1))x^*_{\pi^*(i)})) \\
&=\sum_{e\in\sset{E}}y^*_e. \\
\end{align*}
Note that the last equation is from the fact that $y^*_e$ is equal to the minimum of $\sum_{i=1}^{|e|}(f_e(i)-f_e(i-1))x^*_{\pi(i)}$ for any permutation $\pi\in \sset{P}_e$.

Hence, we have $\sum_{e\in\sset{E}}y^*_e <\gamma^*\sum_{v\in V}{x^*_v}$ by assumption. However, this contradicts the condition that $\gamma^*$ is optimal value.
Therefore, we conclude that we can definitely find a $r\in [0,1]$ such that $\frac{F(S_r)}{|S_r|}\ge \gamma^*$.
\end{proof}

By the proof of Theorem \ref{thm_core_lp}, we can obtain the optimal solution to GDSH from the optimal LP solution  $(\mathbf{x}^*,\mathbf{y}^*)$ by simply solving $max_{r\in[0,1]}\frac{F(S_r)}{|S_r|}$. 
The number of inequalities in LP-GDSH is $O(mr!)$, but this linear program can be still solved in polynomial time  because Inequality \ref{ineq1} can be efficiently separated by Lemma \ref{lemma_ineq_minimum}.

\noindent \textbf{Remark.} It is clear that LP-GDSH generalizes Charikar's linear program \cite{charikar2000greedy} for DSG and Hu's linear program \cite{hu2017maintaining} for DSH (Problem \ref{problem_old}). A very recent work in \cite{chekuri2022densest} showed that the linear program technique can be also used for solving the densest supermodular subset  problem which maximizes a supermodular set function of $S$ over $|S|$. Our GDSH can be a special case of this problem as $F(S)$ is supermodular if $f_e$ is convex. It is also interested to see that our LP-GDSH can be reduced to their linear program by summarizing Inequality \ref{ineq1} over all $e\in \sset{E}$.

\subsection{A Network Flow Algorithm}
\label{sec_flow}
In this section, we introduce a parametric network flow-based approach algorithm, \emph{GDSH-Flow}, for solving GDSH when $f_e$ is non-decreasing convex. 
\emph{GDSH-Flow} is a standard binary search algorithm which finds the optimal density within range $[lb,ub]$. 
Initially, $lb=0$ and $ub=\Psi$. 
\emph{GDSH-Flow} testifies if there is a sub-hypergraph of density $\lambda=\frac{lb+ub}{2}$ by computing $\min_{S\subseteq V}(\gamma|S|-F(S))$.
If  $\min_{S\subseteq V} (\gamma|S|-F(S))\le 0$, there exists a sub-hypergraph of density $\lambda$, then we set $lb$ as $\lambda$. 
Otherwise, it indicates that $\lambda$ is larger than the optimal, we then decrease $ub$ to $\lambda$.
In order to compute $\min_{S\subseteq V} (\gamma|S|-F(S))$ for any $\lambda$, we make use of the minimum cut from a directed network flow $\gr{G}=(U,\sset{A}, \lambda)$ where $U$ and $A$ are vertex set and arc set, respectively.

\begin{algorithm}[ht!]
\DontPrintSemicolon
	\caption{Exact network flow algorithm for GDSH.}
	\label{alg_netflow}
    GDSH-Flow$(\gr{H})$\\
    \Begin{
        $lb\gets 0$,$ub\gets \Psi$ \\ 
        \While{$ub > lb$}{
            $\lambda \gets \frac{lb+ub}{2}$\\
            Build directed flow network $\gr{G}=(U,\sset{A}, \lambda)$ \\
            \If{the cost of min-cut $(X,Y)$ in $\gr{G}$ is larger than $\Psi$}{
                $ub\gets \lambda$   \\
            }\Else{
                $lb\gets \lambda$\\
            }
        }
        build directed flow network $\gr{G}=(U,\sset{A}, lb)$ \\
        compute minimum cut $(X,Y)$ from $\gr{G}$\\
        \Return{$X\cap V$}
    }
\end{algorithm}

To illustrate how to build $\gr{G}=(U,\sset{A}, \lambda)$, we need to first assume that for any $e\in \sset{E}$, $f_e(i)$ returns integers for $i=0,...,|e|$. 
This restriction does not impose any loss of generality because we can always obtain integer values by simultaneously scaling the edge-weight functions with an enough large value $M$ which is a multiple of $10$.
Then, $\gr{G}$ is built by the following steps.

\begin{itemize}
  \item Build a source $s$, a sink $t$ in $\gr{G}$, and make a copy of every vertex of $V$ in $\gr{G}$.
  \item For a vertex $v\in U\setminus \{s,t\}$, add an arc $(v,t)$ with capacity $\lambda$ to $\gr{G}$.
  \item For a hyperedge $e$ in $\gr{H}$, assume $e=\{v_0,...,v_{|e|-1}\}$. Then, add $|e|$ vertices $u^e_0,...,u^e_{|e|-1}$ to $\gr{G}$. Also, add the following arcs to $\gr{G}$.
    \begin{itemize}
      \item For each $i=0,...,|e|-1$, add an arc $(s, u^e_i)$ with capacity $(|e|-i)\alpha^e_i$.
      \item For each $i=0,...,{|e|-1}$, $j=0,...,{|e|-1}$, add an arc $(u^e_i,v_j)$ with capacity $\alpha^e_i$ that
      \[
\alpha^e_i =\left\{
\begin{array}{ll}
f_e(1)-f_e(0) \mbox{ if } i=0, \\
f_e(i+1) + f_e(i-1) - 2f_e(i) \mbox{ if } 0<i<|e|.
\end{array}
\right.
\]
    \end{itemize}
\end{itemize}

\begin{figure}[htb]
  \centering
  \includegraphics[scale=0.7]{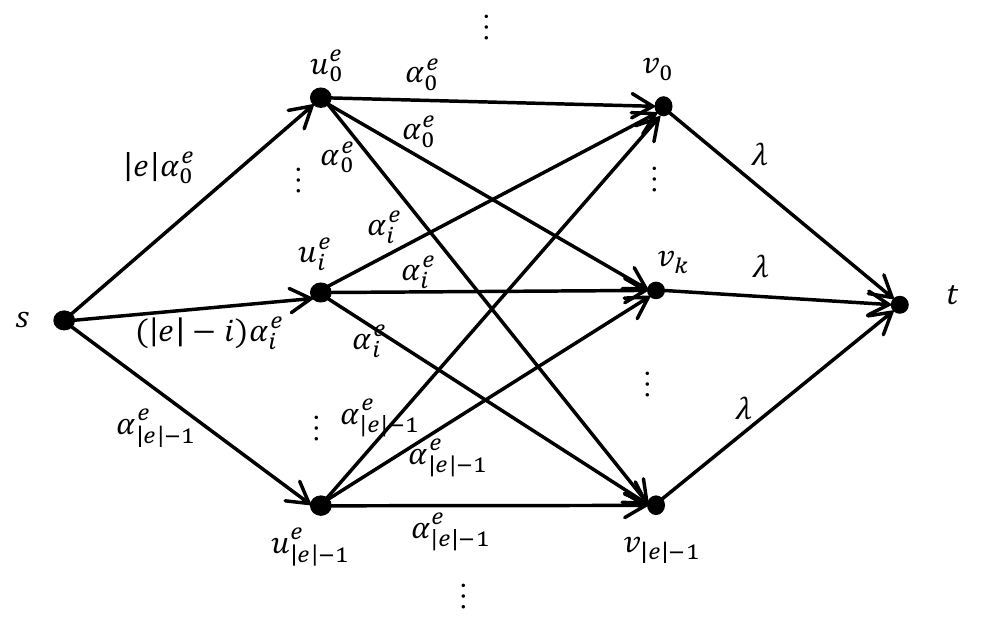}\\
  \caption{An example of directed network $\gr{G}$.}\label{fig_network}
\end{figure}

An illustrative example of $\gr{G}$ is shown in Fig. \ref{fig_network}. Clearly, $\alpha^e_i\ge 0$ for any integer $0\le i<|e|$ because $f_e$ is non-decreasing convex. We have the following statement for $\gr{G}$.
\begin{lemma}
\label{lemma-mincut}
Let $(X,Y)$ be a minimum $s,t$-cut in network $\gr{G}=(U,A,\lambda)$ such that $s\in X$ and $t\in Y$. Denote $S=X\cap V$. Then, the cost of $(X,Y)$ is equal to $\Psi+\lambda|S|-F(S)$.
\end{lemma}
\begin{proof}
 First, by the definition of $a^e_i$, we have,
\begin{equation*}
\small
\begin{aligned}
\sum_{i=0}^{j}{\alpha_i^e} &= f_e(1)-f_e(0) +\sum_{i=1}^{i=j}{((f_e(i+1)-f_e(i))-(f_e(i)-f_e(i-1)))} \\
&= f_e(j+1)-f_e(j)
\end{aligned}
\end{equation*}
for any integer $j<|e|$.

Second, in network $\gr{G}$, $(X,Y)$ is a minimum $s,t$-cut, $S=X \cap V$ and an edge $e=\{v_0,...,v_{|e|-1}\}$. Then, if $|S\cap e|>i$, then $u^e_i\in X$, otherwise, $(X,Y)$ is not a minimum $s,t$-cut. In contrary, if $|S\cap e|<i$, then $u^e_i\in Y$. If $|S\cap e| = i$, $u^e_i$ can be either in $X$ or $Y$ without changing the cost of cut $(X,Y)$.

Let us denote $s_e = |e\cap S|$ for simplicity. With the above observations, we finally get the cost of cut $(X,Y)$ as
\begin{equation*}
\small
\begin{aligned}
&\sum_{e\in \sset{E}}((|e|-s_e)\sum_{i=0}^{s_e-1}{a^e_i} +\sum_{i=s_e}^{|e|-1}{(|e|-i)a^e_i} ) + \lambda |S| \\
=&\sum_{e\in \sset{E}}( \sum_{i=0}^{s_e}a^e_i+...+\sum_{i=0}^{|e|-1}a^e_i) + \lambda |S| \\
=&\sum_{e\in \sset{E}}( f_e(s_e+1)-f_e(s_e)+...+ f_e(|e|)-f_e(|e|-1))+\lambda |S|  \\
=&\sum_{e\in \sset{E}}( f_e(|e|)-f_e(s_e)) + \lambda |S| \\
=&\sum_{e\in \sset{E}}f_e(|e|)-F(S)+\lambda|S|
\end{aligned}
\end{equation*},
which ends the proof.
\end{proof}


For a given $\lambda \ge 0$, Lemma \ref{lemma-mincut} indicates that the cost of minimum $s,t$-cut $(X,Y)$ is $\Psi + \min_{S\subseteq V}(\lambda |S|-F(S))$. Thus, we can decide whether there exits an $S\subseteq V$ such that $\lambda |S|-F(S)<0$ by checking whether the cost of minimum $s,t$-cut of $\gr{G}$ is smaller than $\Psi$.
Therefore, the correctness of Alg. \ref{alg_netflow} is straightforward due to Lemma \ref{lemma-mincut}. 

\begin{theorem}
If $\forall e\in \sset{E}$ in $\gr{H}=(V,\sset{E}, f)$, $f_e$ is a non-decreasing convex function,
then Alg. \ref{alg_netflow} solves GDSH in $O(mincut(p,pr)\log\Psi)$ time where $mincut(N,M)$ is the time of finding minimum $s,t$-cut from a directed flow graph with $N$ vertices and $M$ edges.
\end{theorem}

For any parameter $\lambda$, the number of vertices and edges in  $\gr{G}=(U,\sset{A},\lambda)$ is $n+p+2$ and $n+p+pr$, respectively.
Therefore, the running time of this flow based algorithm is $O(mincut(p,pr)\log\Psi)$. 
For example, if we use the minimum $s,t$-cut algorithm in \cite{goldberg1988new}, which has running time $O(NM\log{\frac{N^2}{M}})$  and space $O(M)$, the flow based algorithm runs in time $O(npr\log(\frac{(n+p)^2}{n+p+pr})\log\Psi)$ and  space $O(pr)$.

\noindent \textbf{Remark.} Readers who are familiar with submodular optimization can realize that  $h_\lambda(S)$ is monotone submodular when $f_e$ is non-decreasing convex for any $e\in \sset{E}$. Therefore $\min_{S\subseteq V}h_\lambda(S)$ can be also solved via \emph{Submodular Function Minimization} algorithms. The  best-known  submodular function minimization algorithm runs in time $O(N^3(\log^2N)EO+N^4\polylog(N))$ where $N$ is the number of elements and $EO$ is the maximum time of evaluating the submodular function \cite{lee2015faster}. In our case, $N=n$ and $EO=O(p)$, the overall time  is $O((n^3p\log^2n+n^4\polylog(n))\log\Psi)$ which is not as efficient as our the network flow based approach.

\subsubsection{Further removing the $\log\Psi$ factor}
Inspired by the technique in \cite{kawase2018densest}, we can obtain an algorithm with time polynomial to the size of input graph and $\epsilon^{-1}$ by a little modification of \emph{GDSH-Flow}. The algorithm, as shown in Alg. \ref{alf_netflow_epsilon}, is named \emph{GDSH-Flow-$\epsilon$}, which is $(1-\epsilon)$ approximation. 

\begin{algorithm}[ht!]
\DontPrintSemicolon
	\caption{$(1-\epsilon)$ approximation algorithm for GDSH.}
	\label{alf_netflow_epsilon}
    GDSH-Flow-$\epsilon$$(\gr{H})$\\
    \Begin{
        Let $e_m=argmax_{e\in \sset{E}}{f_e(|e|)}$\\
        $lb\gets \frac{f_{e_m}(|{e_m}|)}{|e_m|}$,$ub\gets \frac{\Psi}{1}$\\
        \While{$lb<(1-\epsilon)ub$}{
            $\lambda \gets \sqrt{lb*ub}$\\
            Build directed flow network $\gr{G}=(U,\sset{A}, \lambda)$ \\
            \If{the cost of min-cut $(X,Y)$ in $\gr{G}$ is larger than $\sum_{e\in \sset{E}}f_e(|e|)$}{
                $ub \gets \lambda$\\
            }\Else{
                $lb\gets \lambda$\\
            }
        }
        build directed flow network $\gr{G}=(U,\sset{A}, lb)$ \\
        compute minimum cut $(X,Y)$ from $\gr{G}$\\
        \Return{$X\cap V$}
    }
\end{algorithm}

\begin{theorem}
\label{thm-network-appr}
If $\forall e\in \sset{E}$ in $\gr{H}=(V,\sset{E}, f)$, $f_e$ is a non-decreasing convex function, then \emph{GDSH-Flow-$\epsilon$} is a $(1-\epsilon)$ approximation algorithm with running time  $O(mincut(p,pr)\log(\epsilon^{-1}\log(rm)))$ for DHSP.
\end{theorem}
\begin{proof}
The approximation ratio is clearly guaranteed by the stop condition of the algorithm. We mainly show that the number of \emph{while} iterations is bounded by $\log(\epsilon^{-1}\log(rm))$.
The crux is that  $\frac{ub}{lb}$ is shrunk by a square root after every iteration.
Let $i\in {1,...,i^*}$ denote the iteration number of the algorithm, $i^*$ is the last iteration number.
In the $i$th iteration, let $lb^i$ and $ub^i$ be the lower and upper bound respectively.

First, it is clear that $\frac{ub^1}{lb^1}=\frac{|e_m|\Psi}{f_{e_m}(|e_m|)}\le rm$.
Then, we have
\begin{align*}
&\frac{ub^{i+1}}{lb^{i+1}}\le max\left(\frac{ub^{i}}{\lambda^{i}}, \frac{\lambda^{i}}{lb^{i}} \right)  \\
&= max\left(\frac{ub^{i}}{\sqrt{lb^i*ub^i}}, \frac{\sqrt{lb^i*ub^i}}{lb^{i}} \right)  \\
& = \sqrt{\frac{ub^i}{lb^i}}
\end{align*}
Hence, we have $\frac{ub^{i+1}}{lb^{i+1}}\le (\frac{ub^1}{lb^1})^{\frac{1}{2^i}}$.
On the other hand, we have $\frac{ub^{i^*}}{lb^{i*}}\le \frac{1}{1-\epsilon}$.
Therefore, $i^* \in O(\log(\frac{\log(rm)}{\log(\frac{1}{1-\epsilon})}))$.
As $\lim_{\epsilon\rightarrow 0}{\frac{\epsilon}{\log(\frac{1}{1-\epsilon})}}=1$, we have  $i^*\in O(\log(\epsilon^{-1}\log(rm)))$.
Therefore, the the overall running time of is bounded by $O(mincut(p,np)\log(\epsilon^{-1}\log(rm)))$.
\end{proof}


\subsection{A Fast $\frac{1}{r}$-approximation Algorithm}
\label{approx}

We further introduce \emph{GDSH-Approx} in Alg. \ref{alg_approx} to approximate GDSH when all the edge-weight functions are non-decreasing convex.
By a little sacrifice on the accuracy, \emph{GDSH-Approx} is much faster than the above approaches.

\begin{algorithm}[ht!]
\DontPrintSemicolon
	\caption{Find the densest hyper-subgraph approximately.}
	\label{alg_approx}
    \emph{GDSH-Approx}$(\gr{H})$\\
    \Begin{
        $S \gets V$, $S' \gets V$ \\
        \For{$S\neq \emptyset$}{
            Find $v \in \argmin_{v\in S}(F(S)-F(S\setminus \{v\}))$\\
            $S\gets S\setminus \{v\}$ \\
            \If {$\frac{F(S)}{|S|} > \frac{F(S')}{|S'|}$}{
                $S'\gets S$
            }
        }
        \Return{$S'$}
    }
\end{algorithm}
\emph{GDSH-Approx} maintains a subset of vertices $S$. 
In each of the consequent iterations, \emph{GDSH-Approx} identifies $v$, a vertex by which is removed from $S$, the decrease to the total edge-weight of the sub-hypergraph induced by $S$ is minimized.
The algorithm starts with $S=V$ and stops when $S$ becomes empty.
Of all the sets $S$ during the iterations, the one maximizing $\frac{F(S)}{|S|}$ is returned.
To shown the approximate ratio of \emph{GDSH-Approx}, we first need the following observation.

\begin{lemma}
\label{lemma-degree}
Given any $S\subseteq V$ in hypergraph $\gr{H}=(V,\sset{E},f)$, $F(S)\ge \frac{1}{r}\sum_{u\in S}(F(S)-F(S\setminus\{u\}))$.
\end{lemma}
\begin{proof}
For any hyperedge $e$, it is clear that if vertex $u\in e\cap S$, $f_e(|e\cap S|)-f_e(|e\cap (S\setminus \{u\})|) \le  f(|e\cap S|)$  and if $u\notin e\cap S$, $f_e(|e\cap S|)-f_e(|e\cap (S\setminus \{u\})|) =0$. Hence, the following inequality holds.
\[
\begin{array}{ll}
\sum_{u\in S}{(f_e(|e\cap S|)-f_e(|e\cap (S\setminus \{u\})|))}\le |S\cap e|{f_e(|e\cap S|)} \le rf_e(|e\cap S|)
\end{array}
\]
By summarizing the above inequalities for all $e\in \sset{E}$, we have
\[
\sum_{u\in S}(F(S)-F(S\setminus\{u\}) \le rF(S)
\]
which completes the proof.
\end{proof}

Then, we have the following result for \emph{GDSH-Approx}.
\begin{theorem}
\label{thm-approx}
If $\forall e\in \sset{E}$ in $\gr{H}=(V,\sset{E}, f)$, $f_e$ is a non-decreasing convex function, \emph{GDSH-Approx} is a $\frac{1}{r}$-approximate algorithm.
\end{theorem}
\begin{proof}
Assume $S^*\subseteq V$ is a set of density $\lambda^*$ in $\gr{H}$. Due to the optimality of $S^*$, for any $v\in S^*$,
\begin{equation*}
    \lambda^*=\frac{F(S^*)}{|S^*|} \ge \frac{F(S^*\setminus \{v\}))}{|S^*|-1} =\frac{F(S^*) - (F(S^*)-F(S^*\setminus\{v\})))}{|S^*|-1}.
\end{equation*}
With simple elementary transformations of the above inequality, we have $F(S^*)-F(S^*\setminus\{v\}) \ge \lambda^*$.

Now, let us consider the iteration of \emph{GDSH-Approx} before the first vertex  of $S^*$, say $v$, is removed. Call the current set of this iteration $S'$. So, $S^*\subseteq S'$. We have
\[
\begin{array}{ll}
\forall u\in S', &F(S')-F(S'\setminus\{u\}) \ge F(S')-F(S'\setminus\{v\}) \\
&\ge F(S^*)-F(S^*\setminus\{v\}) \ge \lambda^*
\end{array}
\]
where the first inequality follows from greedy strategy in the algorithm and the second inequality follows from the supermodularity of $F(S)$ (since all edge-weight functions are convex). Now, combining Lemma \ref{lemma-degree}, we conclude that
\begin{equation*}
\small
\begin{aligned}
   F(S') &\ge \frac{1}{r}\sum_{u\in S'}(F(S')-F(S'\setminus\{u\})) \ge \sum_{u\in {S'}}\frac{\lambda^*}{r} = \frac{|S'|\lambda^*}{r}\\
\end{aligned}
\end{equation*}
Therefore, $\frac{F(S')}{|S'|}\ge \frac{\lambda^*}{r}$. Since the algorithm returns a set of maximum   density of all the iterations, the approximation ratio follows.
\end{proof}

The number of iterations of Alg. \ref{alg_approx}  is $n$, the time to evaluate $F(S)$ is $O(p)$ in each iteration. Therefore, the running time of a simple implementation of this algorithm can be $O(n^2p)$. Using a minimum-heap to \cite{cormen2009introduction} to maintain the vertices in $S$, we can reduce the time to $O(pr\log{n})$.


\subsubsection{Further reducing the number of iterations}
\label{streaming}

Currently, the number of iterations of \emph{GDSH-Approx} is clearly $\Theta(n)$. 
Motivated by the work in \cite{bahmani2012densest}, we provide a method of revising Alg. \ref{alg_approx} such that the number of iterations reduces to the logarithmic scale. 
The new approximation algorithm is called \emph{GDSH-Para} which is described in Alg. \ref{alg_para}. 
\emph{GDSH-Para} would be very efficient in processing large hypergraphs in the parallel 
processing system because it only have a small number of dependable iterations.

\begin{algorithm}[H]
\DontPrintSemicolon
	\caption{A parallel algorithm to find the densest sub-hypergraph.}
	\label{alg_para}
    \emph{GDSH-Para}$(\gr{H})$\\
    \Begin{
        $S \gets V$, $S' \gets V$ \\
        \For{$S\ne \emptyset$}{
            $\Delta \gets \{v\in S:  F(S)-F(S\setminus\{v\})\le r(1+\epsilon)\frac{F(S)}{|S|}\}$\\
            $S\gets S\setminus \Delta$\\
            \If {$\frac{F(S)}{|S|}>\frac{F(S')}{|S'|}$}{
                $S'\gets S$
            }
        }
        \Return{$S'$}
    }
\end{algorithm}

\begin{theorem}
\label{thm-para}
If $\forall e\in \sset{E}$ in $\gr{H}=(V,\sset{E}, f)$, $f_e$ is a non-decreasing convex function,
then \emph{GDSH-Para} is a $\frac{1}{r(1+\epsilon)}$-approximation with $O(\log_{1+\epsilon}n)$ iterations.
\end{theorem}
\begin{proof}
By the proof of Theorem \ref{thm-approx}, for any vertex $v$ in a optimal solution $S^*$, $F(S^*)-F(S^*\setminus\{v\})\ge \lambda^*$ where $\lambda^*$ is the maximum density. Let us consider the pass before a first vertex $v$ from $S^*$ is removed in the algorithm. Denote the set as $S'$. Similar to the proof of Theorem \ref{thm-approx}, we still have $F(S')-F(S'\setminus\{v\})\ge F(S^*)-F(S^*\setminus\{v\})\ge \lambda^*$ due to the supermodularity of $F$ and the optimality of $S^*$.
Then,
\[
\frac{F(S')}{|S'|}\ge \frac{F(S')-F(S'\setminus\{v\})}{r(1+\epsilon)}\ge \frac{\lambda^*}{r(1+\epsilon)}
\]
where the first inequality is a direct result of the strategy in Line 5 in the algorithm. Hence, we obtain the approximation ratio.

We now estimate the maximum number of iterations. At each iteration, for the current set $S$,
\begin{equation*}
\begin{aligned}
F(S) &\ge \sum_{u\in S}\frac{F(S)-F(S\setminus\{u\})}{r} \\
    &= \sum_{u\in \Delta}\frac{F(S)-F(S\setminus\{u\})}{r} + \sum_{u\in S\setminus \Delta}\frac{F(S)-F(S\setminus\{u\})}{r} \\
    &> r(1+\epsilon)\frac{F(S)}{|S|} \cdot\frac{|S\setminus \Delta|}{r} \\
\end{aligned}
\end{equation*}
 where the first inequality follows from Lemma \ref{lemma-degree}, the second follows from the fact that any $u\in S\setminus \Delta$ satisfies $F(S)-F(S\setminus\{u\}) >  r(1+\epsilon)\frac{F(S)}{|S|}$.
Thus, $|S\setminus \Delta|<\frac{1}{1+\epsilon}|S|$, indicating that the size of $S$ decreases by a factor at least $\frac{1}{1+\epsilon}$ during each iteration. Therefore, the algorithm stops in $O(\log_{1+\epsilon}n)$ iterations.
\end{proof}

Like GDSH-Approx, the space consumption of \emph{GDSH-Para} is $O(n)$ if the minimum-heap data structure is used.

\section{Non-convex Edge-weight Functions}

In this section, we investigate GDSH when some of the edge-weight functions are not non-decreasing convex.

\begin{theorem}
\label{thm-allconcave}
If $\forall e\in \sset{E}$ in $\gr{H}=(V,\sset{E}, f)$, $f_e$ is a non-decreasing concave function, the solution of \emph{GDSH-Approx} is $\{v\}$ where $v=\argmax_{u\in V}{\sum_{e\in \sset{E}:u\in e}f_e(1)}$.
\end{theorem}
\begin{proof}
If for every $e\in \sset{E}$, $f_e$ is non-decreasing concave, then $F(S)$ is a monotone submodular function. Besides, $F(\emptyset)=0$ because $f_e(0)=0$.
We first claim that, for any unit vertex set $S$ that $|S|=1$,  $S\subseteq T\subseteq V$, $F(S)\ge \frac{F(T)}{|T|}$ holds.
To verify the claim, let us assume  $S=\{v_1\}$ and $T=\{v_1, ...,v_p\}$ without loss of generality, where $p\ge 1$ is the size of $T$. 
By submodularity, we have
\[
\begin{array}{ll}
F(T)-F(\{v_1,...,v_{p-1}\}) \le F(S)-F(\emptyset) \\
F(\{v_1,...,v_{p-1}\})-F(\{v_1,...,v_{p-2}\}) \le F(S)-F(\emptyset) \\
\ldots\\
F(\{v_1\})-F(\emptyset) \le F(S)-F(\emptyset)\\
\end{array}
\]
By adding up  the above inequalities, we obtain $F(T)-F(\emptyset) \le p(F(S)-F(\emptyset))$.
As $F(\emptyset)=0$,we have $F(S)\ge \frac{F(T)}{p}= \frac{F(T)}{|T|}$

Now, it is not hard to see that the optimal solution to GDSH is a set with one vertex $v=\argmax_{u\in V}{\sum_{e\in \sset{E}:u\in e}f_e(1)}$.

\end{proof}

On the other hand,  if there are some edge-weight function $f_e$ that is (non-monotonic) concave in $\gr{H}=(V,\sset{E}, f)$, then we have the following NP-hardness result.
\begin{theorem}
\label{thm-nph-subm}
Given a hypergraph $\gr{H}=(V,\sset{E},f)$, if for some $e\in \sset{E}$ $f_e$ is concave and for other $e\in\sset{E}$, $f_e$ is non-decreasing convex, then GDSH is NP-hard.
\end{theorem}
\begin{proof}
We reduce the well-known NP-hard problem, max-cut, to GDSH that edge-weight functions contain both convex and concave functions.

Given an unweighted graph $\gr{H}=(V, \sset{E})$ where $n=|V|$,  the max-cut problem asks to find $T\subseteq V$ such that $cut_{\gr{H}}(T) = |\{e \in \sset{E} : |e\cap T| = 1\}|$ is maximized. To show the reduction, we build an edge-weighted hypergraph $\gr{H}^* = (V^*, \sset{E}^*, f^*)$ which includes both concave and convex edge-weight functions.
\begin{itemize}
  \item Make two disjoint copies of $\gr{H}$ of the same vertex and (hyper)edge sets. Denote the two copies as $\gr{H}' = (V', \sset{E}', f')$ and $\gr{H}'' = (V'', \sset{E}'', f'')$. For each hyperedge $e$ in both $\gr{H}'$ and $\gr{H}''$, set $f_e(i) = 1$ if $i = 1$ and $f_e(i) = 0$ for all $i\neq 1$.
  \item For vertex $v\in V$, insert a hyperedge $e_v = \{v', v''\}$ where $v'$ and $v''$ are the two copies of $v$ in $\gr{H}'$ and $\gr{H}''$, respectively. For each hyperedge $e_v$, set $f_{e_v}(i) = n^2$ if $i = 1$ and $f_{e_v}(i) = 0$ for all $i\neq 1$. Denote  the set of these hyperedges as $\sset{E}'''$.
  \item Add a hyperedge $e_n=V'\cup V''$ and assign the edge-weight function
\[f_{e_n}(i) =\left\{
\begin{array}{ll}
0,  \mbox{ if } 0\le i < n \\
n^2(i-n+1),  \mbox{ if } n\le i\le 2n.
\end{array}
\right.
\]
\end{itemize}
to hyperedge $e_n$.

In summary, $\gr{H}^*$ includes a set of $2n$ vertices and four sets of hyperedges, $\sset{E}'$, $\sset{E''}$, $\sset{E'''}$ and $\{e_n\}$.
The edge-weights functions of hyperedges in $\sset{E}'$, $\sset{E''}$, $\sset{E'''}$ are concave but $f_{e_n}$ is convex. Given $\gr{H}^*$, GDSH is to find a set $S \subseteq V^*$ such that $\frac{F(S)}{|S|}=\frac{{cut_{\gr{H'}}(S) + cut_{\gr{H''}}(S)} + n^2|\{e\in \sset{E}''':|e\cap S|=1\}| + f_{e_n}(|S|)}{|S|}$ is maximized. (To be precise, $cut_{\gr{H'}}(S)$ and $cut_{\gr{H''}}(S)$ are the abbreviations of $cut_{H'}(S\cap V')$ and $cut_{H''}(S\cap V'')$, respectively.)

Suppose that $S^*$ is an optimal solution of GDSH in $H^*$. We first demonstrate that the size of $S^*$ is equal to $n$.
Assume $|S^*|<n$.
\begin{equation*}
\small
\begin{aligned}
\frac{F(S^*)}{|S^*|}&= \frac{cut_{\gr{H'}}(S^*) + cut_{\gr{H''}}(S^*) + n^2|\{e\in \sset{E}''':|e\cap S^*|=1\}| + 0}{|S^*|} \\
       &\le \frac{cut_{\gr{H'}}(S^*) + cut_{\gr{H''}}(S^*) + n^2|S^*|}{|S^*|} \\
       &\le \frac{cut_{\gr{H'}}(S^*) + cut_{\gr{H''}}(S^*)}{|S^*|} + n^2\\
       &< \frac{n\cdot |S^*|}{|S^*|} + n^2\\
       &=n+n^2
\end{aligned}
\end{equation*}
On the other hand, if $|S^*|=n$, $\frac{F(S^*)}{|S^*|} =n^2+n+ \frac{cut_{\gr{H'}}(S^*) + cut_{\gr{H''}}(S^*)}{n}$ which is not smaller than $n+n^2$. Hence, $|S^*|\ge n$ .

Now suppose $|S^*| > n$. Then,
\begin{equation*}
\small
\begin{aligned}
\frac{F(S^*)}{|S^*|} &= \frac{cut_{\gr{H'}}(S^*) + cut_{\gr{H''}}(S^*) + n^2|\{e\in \sset{E}''':|e\cap S^*|=1\}| + n^2(|S^*|-n+1)}{|S^*|} \\
       &\le \frac{cut_{\gr{H'}}(S^*) + cut_{\gr{H''}}(S^*) + n^2(2n-|S^*|)+n^2(|S^*|-n+1)}{|S^*|} \\
       &< \frac{n|S^*|}{|S^*|} + \frac{n^2(n+1)}{|S^*|}\\
       &\le n+n^2
\end{aligned}
\end{equation*}
Clearly, the density when $|S^*|>n$ is still smaller than the density when $|S^*|=n$. Therefore, the size of optimal solution $S^*$ is $n$.

Now, we show that for any vertex $u\in V$, the two copies $u',u''\in V^*$ satisfy either $u'\in S^*$ and $u''\notin S^*$ or $u'\notin S^*$ and $u''\in S^*$.
 Suppose that $u'$ and $u''$ are both in $S^*$. Then for hyperedge $e = \{u', u''\}$, $f_e(|S^*\cap e|) = 0$. By removing $u'$(or $u''$) from $S^*$, we can get a larger density for set $S^*$, which contradicts the fact that $S^*$ is optimal. If we assume neither $u'$ nor $u''$ in $S^*$, we can also find a contradiction by adding $u'$(or $u''$) to $S^*$. Thus, either $u'$ or $u''$ is in $S^*$ but not both.

With the above two properties of the optimal solution $S^*$, we can state that the optimal solution $\frac{F(S^*)}{|S^*|} =n^2+n+\frac{2\max_{S\subseteq V}cut_{\gr{H}}(S)}{n}$. Therefore, the max-cut problem can be reduced to GDSH where the input graph includes both convex and concave edge-weight functions.
\end{proof}

\section{Conclusion}
It is known that the (edge-weighted) densest sub-hypergraph problem is important in many data-mining applications. 
In this paper, we studied this problem with respect to different properties of the edge weight functions and formalized the Generalized Densest Sub-Hypergraph problem (GDSH).
We show that GDSH with non-decreasing convex edge-weight functions can be solved efficiently by a linear program-based approach, a network flow-based approach and a fast greedy approximation algorithm. We also investigated GDSH for some other cases where edge-weight function are not always non-decreasing convex.

In the future, it it would be interesting to extend the study from multiple dimensions.
First, one could consider more properties about the edge weight functions like submodularity, or, one could also investigate faster algorithms when the edge-weight functions are identical. 
Besides, the GDSH problem under some constraints wold be another interesting topic. 
For example, the problem of finding densest subgraph with at least $k$ vertices is NP-hard, but 2-approximated was given in \cite{khuller2009finding}. So, it could be possible to investigate the GDSH with different size constraint.


%
%
%
\bibliographystyle{plain}
\bibliography{hyper}

\end{document}